\documentclass[11pt]{myclass-arXiv}
\usepackage{euscript}
\usepackage{verbatim}
\usepackage[normalem]{ulem}
\usepackage[section]{algorithm}
\usepackage{algorithmic}
\usepackage{paralist}

\newcommand{\eps}{\varepsilon}

\newcommand{\denselist}{\vspace{-.1in} \itemsep -2pt\parsep=-1pt\partopsep -2pt}

\renewcommand{\c}[1]{\ensuremath{\EuScript{#1}}}

\renewcommand{\b}[1]{\ensuremath{\mathbb{#1}}}
\newcommand{\IP}[2]{\left\langle #1, #2 \right\rangle}
\newcommand{\diam}[1]{\ensuremath{\textsf{diam(}#1\textsf{)}}}
\newcommand{\poly}[1]{\ensuremath{\textsf{poly}(#1)}}

\newcommand{\SO}[1]{\ensuremath{\textsf{SO(}#1\textsf{)}}}

\title{Comparing Distributions and Shapes using the Kernel Distance}
\author{Sarang Joshi \\ {\small\textsl{sjoshi@sci.utah.edu}} \and Raj
  Varma Kommaraju \\ {\small\textsl{rajvarma@cs.utah.edu}} \and Jeff
    M. Phillips \\ {\small\textsl{jeffp@cs.utah.edu}}\and Suresh
  Venkatasubramanian \\{\small\textsl{suresh@cs.utah.edu}}}
\date{}
\begin{document}
\pagestyle{empty}
\begin{titlepage}
\maketitle 
\begin{abstract}
Starting with a similarity function between objects, it is possible to define a distance metric on pairs of objects, and more generally on probability distributions over them. These distance metrics have a deep basis in functional analysis, measure theory and geometric measure theory, and have a rich structure that includes an isometric embedding into a (possibly infinite dimensional) Hilbert space.
They have recently been applied to numerous problems in machine learning and shape analysis. 

In this paper, we provide the first algorithmic analysis of these distance metrics.  Our
main contributions are as follows:
\begin{inparaenum}[(i)]
\item We present fast approximation algorithms for computing the kernel distance between two point sets $\c{P}$ and $\c{Q}$ that runs in near-linear time in the size of $\c{P} \cup \c{Q}$ (note that an explicit calculation would take quadratic time).
\item We present polynomial-time algorithms for approximately minimizing the kernel distance under rigid transformation; they run in time $O(n + \text{poly}(1/\eps, \log n))$. 
\item We provide several general techniques for reducing complex objects to convenient sparse representations (specifically to point sets or sets of points sets) which approximately preserve the kernel distance.  In particular, this allows us to reduce problems of computing the kernel distance between various types of objects such as curves, surfaces, and distributions to computing the kernel distance between point sets.  These take advantage of the reproducing kernel Hilbert space and a new relation linking binary range spaces to continuous range spaces with bounded fat-shattering dimension.  
\end{inparaenum}


 \end{abstract}

\end{titlepage}
\pagestyle{plain}

\section{Introduction}
\label{sec:introduction}

Let $K : \reals^d \times \reals^d \to R$ be a kernel function, such as a Gaussian kernel; $K(p,q)$ describes how similar two points $p,q \in \reals^d$ are.  For point sets $\c{P}, \c{Q}$ we can define a similarity function $\kappa(\c{P},\c{Q}) = \sum_{p \in \c{P}} \sum_{q \in \c{Q}} K(p,q)$.  Then the \emph{kernel distance} is defined as
\begin{equation}
D_K(\c{P},\c{Q}) = \sqrt{\kappa(\c{P},\c{P}) + \kappa(\c{Q},\c{Q}) - 2 \kappa(\c{P},\c{Q})}.
\label{eq:CD-def}
\end{equation}
By altering the kernel $K$, and the weighting of elements in $\kappa$, the kernel distance can capture distance between distributions, curves, surfaces, and even more general objects.  


\paragraph{Motivation.}
The earthmover distance (EMD) takes a metric space and two probability distributions over the space, and computes the amount of work needed to "transport" mass from one distribution to another. It has become a metric of choice in computer vision, where images are represented as intensity distributions over a Euclidean grid of pixels. It has also been applied to shape comparison~\cite{knauer}, where shapes are represented by point clouds (discrete distributions) in space. 

While the EMD is a popular way of comparing distributions over a metric space, it is also an expensive one. Computing the EMD requires solving an optimal transportation problem via the Hungarian algorithm, and while approximations exist for restricted cases like the Euclidean plane, they are still expensive, requiring either quadratic time in the size of the input or achieving at best a constant factor approximation when taking linear time~\cite{SJ08, ADIW09}. Further, it is hard to index structures using the EMD for performing near-neighbor, clustering and other data analysis operations. Indeed, there are lower bounds on our ability to embed the EMD into well-behaved normed spaces~\cite{DBLP:conf/soda/AndoniIK08}. 

The kernel distance has thus become an effective alternative to comparing distributions on a metric space. In machine learning, the kernel distance has been used to build metrics on distributions~\cite{Suq95,HB05,hilbert,smola,muller97} and to learn hidden Markov models~\cite{SBSGS10}. In the realm of shape analysis~\cite{Vaillant2005,glaunesthesis,GlaunesJoshi:MFCA:06,DBLP:conf/miccai/DurrlemanPTA08}, the kernel distance (referred to there as the \emph{current distance}) has been used to compare shapes, whether they be point sets, curves, or surfaces. 

All of these methods utilize key structural features of the kernel distance. When constructed using a \emph{positive definite}\footnote{A positive definite function generalizes the idea of a positive definite matrix; see Section~\ref{sec:definitions}.} similarity function $K$, the kernel distance can be interpreted through a lifting map $\phi : \reals^d \to \c{H}$ to a reproducing kernel Hilbert space (RKHS), $\c{H}$. This lifting map $\phi$ is isometric; the kernel distance is precisely the distance induced by the Hilbert space ($D_K(\{p\},\{q\}) = \|\phi(p) - \phi(p)\|_{\c{H}}$). 
Furthermore, a point set $\c{P}$ has an isometric lifted representation $\Phi(\c{P}) = \sum_{p \in \c{P}} \phi(p)$ as a single vector in $\c{H}$ so $D_K(\c{P}, \c{Q}) = \|\Phi(\c{P}) - \Phi(\c{Q})\|_\c{H}$.  
Moreover, by choosing an appropriately scaled basis, this becomes a simple $\ell_2$ distance, so all algorithmic tools developed for comparing points and point sets under $\ell_2$ can now be applied to distributions and shapes. 

Dealing with uncertain data provides another reason to study the kernel distance. Rather than thinking of $K(\cdot, \cdot)$ as a similarity function, we can think of it as a way of capturing spatial uncertainty;  $K(p,q)$ is the likelihood that the object claimed to be at $p$ is actually at $q$. For example, setting $K(p,q) = \exp( - \|p-q\|^2/\sigma)/(\sqrt{2\pi} \sigma)$ gives us a Gaussian blur function.  
In such settings, the kernel distance $D^2_K(\c{P},\c{Q})$ computes the symmetric difference $|\c{P} \triangle \c{Q}|$ between shapes with uncertainty described by $K$.

\paragraph{Our work.}
We present the first algorithmic analysis of the kernel distance. Our
main contributions are as follows:
\begin{inparaenum}[(i)]

\item We present fast approximation algorithms for computing the kernel distance between two point sets $\c{P}$ and $\c{Q}$ that runs in near-linear time in the size of $\c{P} \cup \c{Q}$ (note that an explicit calculation would take quadratic time).
\item We present polynomial-time algorithms for approximately minimizing the kernel distance under rigid transformation; they run in time $O(n + \text{poly}(1/\eps, \log n))$. 
\item We provide several general techniques for reducing complex objects to convenient sparse representations (specifically to point sets or sets of points sets) which approximately preserve the kernel distance.  In particular, this allows us to reduce problems of computing the kernel distance between various types of objects such as curves, surfaces, and distributions to computing the kernel distance between point sets.  
\end{inparaenum}

We build these results from two core technical tools. 
The first is a lifting map that maps objects into a finite-dimensional Euclidean space while approximately preserving the kernel distance. We believe that the analysis of lifting maps is of independent interest; indeed, these methods are popular in machine learning~\cite{hilbert,smola,SBSGS10} but (in the realm of kernels) have received less attention in algorithms. 
Our second technical tool is an theorem relating $\eps$-samples of range spaces defined with kernels to standard $\eps$-samples of range spaces on $\{0,1\}$-valued functions.  This gives a simpler algorithm than prior methods in learning theory that make use of the $\gamma$-fat shattering dimension, and yields smaller $\eps$-samples. 

\section{Preliminaries}

\paragraph{Definitions.}
\label{sec:definitions}
For the most general case of the kernel distance (that we will consider in this paper) we associate a unit vector $U(p)$ and a weighting $\mu(p)$ with every $p \in \c{P}$.  Similarly we associate a unit vector $V(p)$ and weighting $\nu(q)$ with every $q \in \c{Q}$.  Then we write 
\begin{equation}
\kappa(\c{P},\c{Q}) = \int_{p \in \c{P}} \int_{q \in \c{Q}} K(p,q) \IP{U(p)}{V(q)} \, d\mu(p) d\nu(q),
\label{eq:kappa}
\end{equation}
where $\langle \cdot, \cdot \rangle$ is the Euclidean inner product.  This becomes a distance $D_K$, defined through (\ref{eq:CD-def}).

When $\c{P}$ is a curve in $\b{R}^d$ we let $U(p)$ be the tangent vector at $p$ and $\mu(p) = 1$.  When $\c{P}$ is a surface in $\b{R}^3$ we let $U(p)$ be the normal vector at $p$ and $\mu(p)=1$.  This can be generalized to higher order surfaces through the machinery of $k$-forms and $k$-vectors~\cite{Vaillant2005,GI}.  

When $\c{P}$ is an arbitrary probability measure\footnote{We avoid the use of the term 'probability distribution' as this conflicts with the notion of a (Schwarz) distribution that itself plays an important role in the underlying theory.} in $\b{R}^d$, then all $U(p)$ are identical unit vectors and $\mu(p)$ is the probability of $p$.  
For discrete probability measures, described by a point set, we replace the integral with a sum and $\mu(p)$ can be used as a weight $\kappa(\c{P},\c{Q}) = \sum_{p \in \c{P}} \sum_{q \in \c{Q}} K(p,q) \mu(p) \nu(q)$.

\paragraph{From distances to metrics.}   
When $K$ is a symmetric similarity function (i.e. $K(p,p) = \max_{q \in \b{R}^d} K(p,q)$, \, $K(p,q) = K(q,p)$, and $K(p,q)$ decreases as $p$ and $q$ become ``less similar'') then $D_K$ (defined through (\ref{eq:kappa}) and (\ref{eq:CD-def})) is a distance function, but may not be a metric.   
However, when $K$ is positive definite, then this is sufficient for $D_K$ to not only be a metric\footnote{Technically this is not completely correct; there are a few special cases, as we will see, where it is a pseudometric~\cite{hilbert}.}, but also for $D^2_K$ to be of negative type~\cite{deza}.  

We say that a symmetric function $K : \reals^d \times \reals^d \rightarrow \reals$ is a  \emph{symmetric positive definite kernel} if for any nonzero $L_2$ function $f$ it satisfies 
$ 
\int_{p \in \c{P}} \int_{q \in \c{Q}} f(q) K(p, q) f(p) \; dp dq > 0.
$
The proof of $D_K$ being a metric follows by considering the reproducing kernel Hilbert space $\c{H}$ associated with such a $K$~\cite{Aronszajn1950}. Moreover, $D_K$ can be expressed very compactly in this space.  The lifting map $\phi : \reals^d \to \c{H}$ associated with $K$ has the ``reproducing property'' $K(p,q) = \langle \phi(p), \phi(q) \rangle_\c{H}$. 
So by linearity of the inner product, $\Phi(\c{P}) = \int_{p \in \c{P}} \phi(p) \, d\mu(p)$ can be used to retrieve $D_K(\c{P}, \c{Q}) = \| \Phi(\c{P}) - \Phi(\c{Q})\|_\c{H}$ using the induced norm $\| \cdot \|_{\c{H}}$ of $\c{H}$.  Observe that this defines a norm $\|\Phi(\c{P})\|_{\c{H}} = \sqrt{\kappa(\c{P},\c{P})}$ for a shape.

\paragraph{Examples.}
If $K$ is the ``trivial'' kernel, where $K(p,p) = 1$ and $K(p,q) = 0$ for $p \ne q$, then the distance between any two sets (without multiplicity) $\c{P}, \c{Q}$ is $D^2_K(\c{P}, \c{Q}) = |\c{P} \Delta \c{Q}|$, where $\c{P} \Delta \c{Q} = \c{P} \cup \c{Q} \setminus (\c{P} \cap \c{Q})$ is the symmetric difference. In general for arbitrary probability measures, $D_K(\c{P}, \c{Q}) = \|\mu - \nu\|_2$. 
If $K(p,q) = \langle p,q \rangle$, the Euclidean dot product, then the resulting lifting map $\phi$ is the identity function, and the distance between two measures is the Euclidean distance between their means, which is a pseudometric but not a metric.

\paragraph{Gaussian properties.}
To simplify much of the presentation of this paper we focus on the case where the kernel $K$ is the Gaussian kernel; that is $K(p,q) = e^{-||p-q||^2/h}$.  Our techniques carry over to more general kernels, although the specific bounds will depend on the kernel being used. 
We now encapsulate some useful properties of Gaussian kernels in the following lemmata.  
When approximating $K(p,q)$, the first allows us to ignore pairs of points further that $\sqrt{h \ln(1/\gamma)}$ apart, the second allows us to approximate the kernel on a grid.
\begin{lemma}[Bounded Tails]
If $||p-q|| > \sqrt{h \ln (1/\gamma)}$ then $K(p,q) < \gamma$.  
\label{lem:G-dist}
\end{lemma}

\begin{lemma}[Lipschitz]
\label{lem:grid-eps}
For $\delta \in \b{R}^d$ where $\|\delta\| < \eps$, for points $p,q \in \b{R}^d$ we have $|K(p,q) - K(p,q+\delta)| \leq \eps/\sqrt{h}$.
\end{lemma}
\begin{proof}
The slope for $\psi(x) = e^{-x^2/h}$ is the function $\psi'(x) = -(2/h)xe^{-x^2/h}$. $\psi'(x)$ is maximized when $x = \sqrt{h/2}$, which yields $\psi'(\sqrt{h/2}) = -\sqrt{2/he} < 1/\sqrt{h}$. Thus $|\psi(x) - \psi(x+\eps)| < \eps/\sqrt{h}$.  And since translating by $\delta$ changes $\|p-q\|$ by at most $\eps$, the lemma holds. 
\end{proof}


\subsection{Problem Transformations}
In prior research employing the kernel distance, \emph{ad hoc} discretizations are used to convert the input objects (whether they be distributions, point clouds, curves or surfaces) to weighted discrete point sets. This process introduces error in the distance computations that usually go unaccounted for. In this subsection, we provide algorithms and analysis for \emph{rigorously} discretizing input objects with guarantees on the resulting error. These algorithms, as a side benefit, provide a formal error-preserving reduction from kernel distance computation on curves and surfaces to the corresponding computations on discrete point sets. 

After this section, we will assume that all data sets considered $\c{P}, \c{Q}$ are discrete point sets of size $n$ with weight functions $\mu : \c{P} \to \reals$ and $\nu : \c{Q} \to \reals$. 
The weights need not sum to $1$ (i.e. need not be \emph{probability} measures), nor be the same for $\c{P}$ and $\c{Q}$; we will set $W = \max(\sum_{p \in \c{P}} \mu(p), \sum_{q \in \c{Q}} \nu(q))$ to denote the total measure.  This implies (since $K(p,p)=1$) that $\kappa(\c{P},\c{P}) \leq W^2$.  
All our algorithms will provide approximations of $\kappa(\c{P},\c{Q})$ within additive error $\eps W^2$. Since without loss of generality we can always normalize so that $W = 1$, our algorithms all provide an additive error of $\eps$.  
We also set $\Delta = (1/h)\max_{u,v \in \c{P} \cup \c{Q}} \|u-v\|$ to capture the normalized diameter of the data.

\paragraph{Reducing orientation to weights.}
\label{sec:from-2-surfaces}
The kernel distance between two oriented curves or surfaces can be reduced to a set of distance computations on appropriately weighted point sets.  
We illustrate this in the case of surfaces in $\reals^3$.  The same construction will also work for curves in $\reals^d$.  

For each point $p \in \c{P}$ we can decompose $U(p) \triangleq (U_1(p), U_2(p), U_3(p))$ into three fixed orthogonal components such as the coordinate axes $\{e_1, e_2, e_3\}$.  Now
\vspace{-.1in}
\begin{eqnarray*}
\kappa(\c{P},\c{Q}) 
= & 
\displaystyle{\int_{p \in \c{P}} \int_{q \in \c{Q}} K(p,q) \IP{U(p)}{V(q)} \, d\mu(p) d\nu(q)}
&=
\int_{p \in \c{P}} \int_{q \in \c{Q}} K(p,q) \sum_{i=1}^3 (U_i(p) V_i(q)) \, d\mu(p) d\nu(q)
\\ =&
\displaystyle{\sum_{i=1}^3 \int_{p \in \c{P}} \int_{q \in \c{Q}} K(p,q) (U_i(p) V_i(q)) \, d\mu(p) d\nu(q)}
&=
\sum_{i=1}^3 \kappa(\c{P}_i, \c{Q}_i),  
\end{eqnarray*}
where each $p \in \c{P}_i$ has measure $\mu_i(p) = \mu(p) \|U_i(p)\|$.  
When the problem specifies $U$ as a unit vector in $\b{R}^d$, this approach reduces to  $d$ independent problems without unit vectors.

\paragraph{Reducing continuous to discrete.}
\label{sec:surface kernel distance}
We now present two simple techniques (gridding and sampling) to reduce a continuous $\c{P}$ to a discrete point set, incurring at most $\eps W^2$ error.

We construct a grid $G_\eps$ (of size $O((\Delta/\eps)^d)$) on a smooth shape $\c{P}$, so no point\footnote{For distributions with decaying but infinite tails, we can truncate to ignore tails such that the integral of the ignored area is at most $(1-\eps/2)W^2$ and proceed with this approach using $\eps/2$ instead of $\eps$.}  $p \in \c{P}$ is further than $\eps \sqrt{h}$ from a point $g \in G_\eps$.  
Let $P_g$ be all points in $\c{P}$ closer to $g$ than any other point in $G_\eps$.  
Each point $g$ is assigned a weight $\mu(g) = \int_{p \in P_g} 1 \, d\mu(p)$.  
The correctness of this technique follows by Lemma \ref{lem:grid-eps}.

Alternatively, we can sample $n = O((1/\eps^2) (d + \log(1/\delta))$ points at random from $\c{P}$.  
If we have not yet reduced the orientation information to weights, we can generate $d$ points each with weight $U_i(p)$.
This works with probability at least $1-\delta$ by invoking a coreset technique summarized in Theorem \ref{thm:random-coreset}.

For the remainder of the paper, we assume our input dataset $\c{P}$ is a weighted point set in $\b{R}^d$ of size $n$.

\section{Computing the Kernel Distance I: WSPDs}
\label{sec:WSPD}

The well-separated pair decomposition (WSPD)~\cite{CK95,HP} is a standard data structure to approximately compute pairwise sums of distances in near-linear time. A consequence of Lemma~\ref{lem:grid-eps} is that we can upper bound the error of estimating $K(p,q)$ by a nearby pair $K(\tilde p,\tilde q)$. Putting these observations together yields (with some work) an approximation for the kernel distance.
Since $D_K^2(\c{P},\c{Q}) =  \kappa(\c{P},\c{P}) +  \kappa(\c{Q},\c{Q}) - 2\kappa(\c{P},\c{Q})$, the problem reduces to computing $\kappa(\c{P},\c{Q})$ efficiently and with an error of at most $(\eps/4)W^2$.  

Two sets $A$ and $B$ are said to be \emph{$\alpha$-separated}~\cite{CK95} if 
$
\max\{\diam{A},\diam{B}\} \leq \alpha \min_{a \in A, b \in B} ||a-b||
$. Let $A \otimes B = \{ \{x,y\} \mid x\in A, y \in B\}$ denote the set of all unordered pairs of elements formed by $A$ and $B$. 
An \emph{$\alpha$-WSPD} of a point set $P$ is a set of pairs $\c{W} = \left\{ \{A_1, B_1\}, \ldots, \{A_s, B_s\}\right\}$ such that 
\begin{itemize}[(i)] \denselist
\item[(i)] $A_i, B_i \subset P$ for all $i$,
\item[(ii)] $A_i \cap B_i = \emptyset$  for all $i$,
\item[(iii)] disjointly $\bigcup_{i=1}^s A_i \otimes B_i = P \otimes P$, and
\item[(iv)] $A_i$ and $B_i$ are $\alpha$-separated for all $i$.  
\end{itemize}
For a point set $P \subset \b{R}^d$ of size $n$, we can construct an $\alpha$-WSPD of size $O(n/\alpha^d)$ in time $O(n \log n + n/\alpha^d)$~\cite{HP,Cla83}.  

We can use the WSPD construction to compute $D_K^2(\c{P},\c{Q})$ as follows.  We first create an $\alpha$-WSPD of $\c{P} \cup \c{Q}$ in $O(n \log n + n/\alpha^d)$ time.  Then for each pair $\{A_i, B_i\}$ we also store four sets $A_{i,\c{P}} = \c{P} \cap A_i$, $A_{i,\c{Q}} = \c{Q} \cap A_i$, $B_{i,\c{P}} = \c{P} \cap B_i$, and $B_{i,\c{Q}} = \c{Q} \cap B_i$.  Let $a_i \in A_i$ and $b_i \in B_i$ be arbitrary elements, and let $D_i = \|a_i - b_i\|$. By construction, $D_i$ approximates the distance between any pair of elements in $A_i \times B_i$ with error at most $2\alpha D_i$.

In each pair $\{A_i, B_i\}$, we can compute the weight of the edges from $\c{P}$ to $\c{Q}$: 
\[
W_i = \bigg(\sum_{p \in A_{i,\c{P}}}\mu(p)\bigg) \bigg(\sum_{q \in B_{i,\c{Q}}}\nu(q)\bigg) + 
\bigg(\sum_{q \in A_{i,\c{Q}}} \nu(q) \bigg) \bigg(\sum_{p \in B_{i,\c{P}}} \mu(p)\bigg).
\]
We estimate the contribution of the edges in pair $(A_i, B_i)$ to $\kappa(\c{P}, \c{Q})$ as 
\[
\sum_{(a,b) \in A_{i,P} \times B_{i,Q}} \mu(a) \nu(b) e^{-D_i^2/h} + \sum_{(a,b) \in A_{i,Q} \times B_{i,P}} \mu(b) \nu(a) e^{-D_i^2/h}
=
W_i e^{-D_i^2 /h}.
\]
Since $D_i$ has error at most $2\alpha D_i$ for each pair of points, Lemma \ref{lem:grid-eps} bounds the error as at most $W_i (2\alpha D_i/\sqrt{h})$.

In order to bound the total error to $(\eps/4) W^2$, we bound the error for each pair by $(\eps/4) W_i$ since $\sum_i W_i = \sum_{p \in P} \sum_{q \in Q} \mu(p) \nu(q) = W^2$.  By Lemma \ref{lem:G-dist}, if $D_i > \sqrt{h \ln(1/\gamma)}$, then $e^{- D_i^2/h} < \gamma$.  So for any pair with $D_i > 2\sqrt{h \ln(4/\eps)}$, (for $\alpha < 1/2$) we can ignore, because they cannot have an effect on $\kappa(\c{P}, \c{Q})$ of more than $(1/4) \eps W_i$, and thus cannot have error more than $(1/4) \eps W_i$.  

Since we can ignore pairs with $D_i > 2 \sqrt{h \ln(4/\eps)}$, each pair will have error at most $W_i(2\alpha(2 \sqrt{h \ln(4/\eps)}/\sqrt{h})$ $=$ $W_i(4\alpha \sqrt{\ln(4/\eps)})$.  We can set this equal to $(\eps/4)W_i$ and solve for $\alpha < (1/4) \eps /\sqrt{\ln(4/\eps)}$.  This ensures that each pair with $D_i \leq 2 \sqrt{h \ln(4/\eps)}$ will have error less than $(\eps/4)W_i$, as desired.  
%

\begin{theorem}
By building and using an $((\eps/4)/\sqrt{\ln(4/\eps)})$-WSPD, we can compute a value $U$ in time $O(n \log n + (n/\eps^d) \log^{d/2}(1/\eps))$, such that 
\[
\left| U - D_K^2(\c{P},\c{Q}) \right| \leq \eps W^2.
\] 
\label{thm:WSPD}
\end{theorem}

\section{Computing the Kernel Distance II: Approximate Feature Maps}

In this section, we describe (approximate) feature representations $\Phi(\c{P}) = \sum_{p \in \c{P}} \phi(p) \mu(p)$ for shapes and distributions that reduce the kernel distance computation to an $\ell_2$ distance calculation $\|\Phi(\c{P}) - \Phi(\c{Q}) \|_\c{H}$ in an RKHS, $\c{H}$.  This mapping immediately yields algorithms for a host of analysis problems on shapes and distributions, by simply applying Euclidean space algorithms to the resulting feature vectors. 

The feature map $\phi$ allows us to translate the kernel distance (and norm) computations into operations in a RKHS that take time $O(n \rho)$ if $\c{H}$ has dimension $\rho$, rather than the brute force time $O(n^2)$. Unfortunately, $\c{H}$ is in general infinite dimensional, including the case of the Gaussian kernel. Thus, we use dimensionality reduction to find an approximate mapping $\tilde{\phi} : \reals^d \to \reals^\rho$ (where $\tilde \Phi(\c{P}) = \sum_{p \in \c{P}} \tilde \phi(p)$) that approximates $\kappa(\c{P}, \c{Q})$:
\[
\bigg| \sum_{p \in P}\sum_{q \in Q} K(p,q)\mu(p)\nu(q) - \sum_{p \in P}\sum_{q \in Q}\IP{\tilde \phi(p)}{\tilde \phi(q)} \bigg| \leq \eps W^2.
\]

The analysis in the existing literature on approximating feature space does not directly bound the dimension $\rho$ required for a specific error bound\footnote{Explicit matrix versions of the Johnson-Lindenstraus lemma~\cite{JL84} cannot be directly applied because the source space is itself infinite dimensional, rather than  $\reals^d$.}.
We derive bounds from two known techniques:  
random projections~\cite{Rahimi2007} (for shift-invariant kernels, includes Gaussians)  
and 
the Fast Gauss Transform~\cite{Yang2003,Greengard1991} (for Gaussian kernel). 
We produce three different features maps, with different bounds on the number of dimensions $\rho$ depending on $\log n$ ($n$ is the number of points), $\eps$ (the error), $\delta$ (the probability of failure), $\Delta$ (the normalized diameter of the points), and/or $d$ (the ambient dimension of the data \emph{before} the map).

\subsection{Random Projections Feature Space}
\label{sec:rahimi-recht-appr}

Rahimi and Recht~\cite{Rahimi2007} proposed a feature mapping that essentially applies an implicit Johnson-Lindenstrauss projection from $\c{H} \to \b{R}^\rho$.
The approach works for any shift invariant kernel (i.e one that can be written as $K(p,q) = k(p-q)$). 
For the Gaussian kernel, $k(z) = e^{-\|z\|^2/2}$, where $z \in \b{R}^d$. 
Let the Fourier transform of $k : \reals^d \to \reals^+$ is $g(\omega) = (2\pi)^{-d/2} e^{-\|\omega\|^2/2}$.  
A basic result in harmonic analysis \cite{bochner} is that $k$ is a kernel if and only if $g$ is a measure (and after scaling, is a probability distribution).  
Let $\omega$ be drawn randomly from the distribution defined by $g$:
\[
k(x-y) 
= 
\int_{\omega \in \reals^d} g(\omega) e^{\iota \IP{\omega}{x-y}} \; d\omega
=
E_\omega[\IP{\psi_{\omega}(x)}{\psi_{\omega}(y)}],
\]
where $\psi_{\omega}(z) = (\cos (\IP{\omega}{z}), \sin(\IP{\omega}{z}))$ are the real and imaginary components of $e^{\iota \IP{\omega}{z}}$.  This implies that $\IP{\psi_\omega(x)}{\psi_\omega(y)}$ is an unbiased estimator of $k(x-y)$. 

We now consider a $\rho$-dimensional feature vector $\phi_{\Upsilon} : P \to \b{R}^{\rho}$ where the $(2i-1)$th and $(2i)$th coordinates are described by $\mu(p) \psi_{\omega_i}(p)/(\rho/2) = (2\mu(p) \cos(\IP{\omega_i}{z})/\rho, 2\mu(p) \sin(\IP{\omega_i}{z})/\rho)$ for some $\omega_i \in \Upsilon = \{\omega_1, \ldots, \omega_{\rho/2}\}$ drawn randomly from $g$.  Next we prove a bound on $\rho$ using this construction.  

\begin{lemma}\label{lem:rand-feat}
For $\phi_\Upsilon : \c{P} \cup \c{Q} \to \reals^\rho$ with $\rho = O((1/\eps^2) \log(n/\delta))$, with probability $\geq 1-\delta$
\[
\bigg|\sum_{p \in \c{P}} \sum_{q \in \c{Q}} K(p,q) \mu(p) \nu(q) - \sum_{p \in \c{P}} \sum_{q  \in \c{Q}} \IP{\phi_\Upsilon(p)}{\phi_\Upsilon(q)} \bigg| \leq \eps W^2.
\]
\end{lemma}
\begin{proof}
We make use of the following Chernoff-Hoeffding bound.   Given a set $\{X_1, \ldots, X_n\}$ of independent random variables, such that $| X_i - E[X_i] | \leq \Lambda$, then for $M = \sum_{i=1}^n X_i$ we can bound $\Pr[| M - E[M] | \geq \alpha] \leq 2 e^{-2\alpha^2 / (n\Lambda^2)}$.  
We can now bound the error of using $\phi_{\Upsilon}$ for any pair $(p,q) \in \c{P} \times \c{Q}$ as follows:

\begin{eqnarray*}
\lefteqn{\Pr\left[\left| \IP{\phi_{\Upsilon}(p)}{\phi_{\Upsilon}(q)} - \mu(p)\nu(q)k(p-q)\right| \geq \eps \mu(p) \nu(q)\right] }  
\\ &= &
\Pr\left[\left| \IP{\phi_{\Upsilon}(p)}{\phi_{\Upsilon}(q)} - E_\Upsilon \left[\IP{\phi_{\Upsilon}(p)}{\phi_{\Upsilon}(q)}\right] \right| \geq \eps \mu(p)\nu(q)\right]
\\ &\leq&
\Pr\left[\left| \sum_i \frac{2}{\rho} \mu(p)\nu(q)\IP{\psi_{\omega_i}(p)}{\psi_{\omega_i}(q)} - E_\Upsilon\left[\sum_i \frac{2}{\rho} \mu(p) \nu(q) \IP{\psi_{\omega_i}(p)}{\psi_{\omega_i}(q)}\right] \right| \geq \eps \mu(p)\nu(q)\right]
\\ &\leq&
2 e^{-2 (\eps \mu(p)\nu(q))^2 / (\rho \Lambda^2/2)}
\leq 
2 e^{- \rho \eps^2/64},
\end{eqnarray*}
where the last inequality follows by $\Lambda \leq 2 \max_{p,q} (2/\rho) \mu(p)\nu(q) \IP{\psi_\omega(p)}{\psi_\omega(q)} \leq 8(2/\rho) \mu(p)\nu(q)$ since for each pair of coordinates $\|\psi_\omega(p)\| \leq 2$ for all $p \in \c{P}$ (or $q \in \c{Q}$).  
By the union bound, the probability that this holds for all pairs of points $(p,q) \in \c{P} \times \c{Q}$ is given by
\[
\Pr\left[\forall_{(p,q) \in \c{P} \times \c{Q}}\left| \IP{\phi_{\Upsilon}(p)}{\phi_{\Upsilon}(q)} - \mu(p)\nu(q)k(p-q)\right| \geq \eps \mu(p) \nu(q)\right]
 \leq (n^2) 2 e^{-\rho \eps^2/64}.
\]
Setting $\delta \geq n^2 2 e^{-\rho\eps^2/64}$ and solving for $\rho$ yields that for $\rho = O((1/\eps^2) \log (n/\delta))$, with probability at least $1-\delta$, for all $(p,q) \in \c{P} \times \c{Q}$ we have $|\mu(p)\nu(q)k(p-q) - \IP{\phi_\Upsilon(x)}{\phi_\Upsilon(y)}| \leq \eps \mu(p)\nu(q)$.  
It follows that with probability at least $1-\delta$
\[
\bigg| \sum_{p \in \c{P}}\sum_{q \in \c{Q}} \mu(p)\nu(q) K(p,q) - \sum_{p \in \c{P}}\sum_{q \in \c{Q}} \IP{\phi_\Upsilon(p)}{\phi_\Upsilon(q)} \bigg| \leq \eps \sum_{p\in \c{P}} \sum_{q \in \c{Q}}\mu(p) \nu(q) \leq \eps W^2.  \qedhere
\]
\end{proof}

Note that the analysis of Rahimi and Recht~\cite{Rahimi2007} is done for unweighted point sets (i.e. $\mu(p) = 1$) and actually goes further, in that it yields a guarantee for any pair of points taken from a manifold $\c{M}$ having diameter $\Delta$. They do this by building an $\eps$-net over the domain and applying the above tail bounds to the $\eps$-net. We can adapt this trick to replace the  $(\log n)$ term in $\rho$ by a $(d \log (\Delta/\eps))$ term, recalling $\Delta = (1/h) \max_{p,p' \in \c{P} \cup \c{Q}} \|p - p'\|$. This leads to the same guarantees as above with a dimension of $\rho = O((d/\eps^2) \log (\Delta/ \eps \delta))$.

\subsection{Fast Gauss Transform Feature Space}
\label{ssec:ifgt}
The above approach works by constructing features in the frequency domain. In what follows, we present an alternative approach that operates in the spatial domain directly. We base our analysis on the Improved Fast Gauss Transform (IFGT)~\cite{Yang2003}, an improvement on the Fast Gauss Transform.  We start with a brief review of the IFGT (see the original work~\cite{Yang2003} for full details). 

\paragraph{IFGT feature space construction.}
The goal of the IFGT is to approximate $\kappa(\c{P},\c{Q})$.  
First we rewrite $\kappa(\c{P},\c{Q})$ as the summation $\sum_{q \in \c{Q}}G(q)$ where $G(q) = \nu(q) \sum_{p \in \c{P}} e^{-\|p-q\|^2/h^2} \mu(p)$.  
Next, we approximate $G(q)$ in two steps.  First we rewrite 
\[
G(q) = \nu(q) \sum_{p \in P} \mu(p) e^{-\frac{\|q - x_*\|^2}{h^2}} e^{-\frac{\|p - x_*\|^2}{h^2}} e^{\frac{2\|q - x_*\| \cdot \|p - x_*\|}{h^2}},
\]
where the quantity $x_*$ is a fixed vector that is usually the centroid of $\c{P}$.  
The first two exponential terms can be computed for each $p$ and $q$ once.  
Second, we approximate the remaining exponential term by its Taylor expansion 
$e^v = \sum_{i\ge 0} \frac{v^i}{i!}$. 
%
After a series of algebraic manipulations, the following expression emerges:
\[ 
G(q) = \nu(q) e^{- \frac{\|q - x_*\|^2}{h^2}}  \sum_{\alpha \ge 0} C_\alpha \Bigl(\frac{q - x_*}{h}\Bigr)^\alpha
\]
where $C_\alpha$ is given by
\[ 
C_\alpha = \frac{2^{|\alpha|}}{\alpha !}\sum_{p \in P} \mu(p) e^{- \frac{\|p - x_*\|^2}{h^2}}\Bigl(\frac{p - x_*}{h}\Bigr)^\alpha .  
  \]
The parameter $\alpha$ is a \emph{multiindex}, and is actually a vector $\alpha = (\alpha_1, \alpha_2, \ldots, \alpha_d)$ of dimension $d$. 
The expression $z^\alpha$, for $z \in \reals^d$, denotes the monomial $z_1^{\alpha_1} z_2^{\alpha_2} \ldots z_d^{\alpha_d}$, the quantity $|\alpha|$ is the \emph{total degree} $\sum \alpha_i$, and the quantity $\alpha! = \Pi_i (\alpha_i !)$. The multiindices are sorted in \emph{graded lexicographic order}, which means that $\alpha$ comes before $\alpha'$ if $|\alpha| < |\alpha'|$, and two multiindices of the same degree are ordered lexicographically.

The above expression for $G(q)$ is an exact infinite sum, and is approximated by truncating the summation at multiindices of total degree $\tau-1$. Note that there are at most $\rho = \binom{\tau+d-1}{d} = O(\tau^d)$ such multiindices. We now construct a mapping $\tilde{\phi} : \reals^d \rightarrow \reals^\rho$. 
Let 
$
\tilde{\phi}(p)_\alpha = \sqrt{\frac{2^{|\alpha|}}{\alpha!}} \mu(p) e^{-\frac{\|p - x_*\|^2}{h^2}}\Bigl(\frac{p - x_*}{h}\Bigr)^\alpha.
$ 
Then 
\[ 
G(q)= \sum_\alpha \tilde{\phi}(q)_\alpha \sum_{p \in P}\tilde{\phi}(p)_\alpha 
\]
and $S = \sum_{q \in Q} G(q)$ is then given by 
\[
 S = 
 \sum_{p \in P}\sum_{q \in Q} \sum_\alpha \tilde{\phi}(q)_\alpha \tilde{\phi}(p)_\alpha  
= 
\sum_{p \in P} \sum_{q \in Q} \langle \tilde{\phi}(q), \tilde{\phi}(p) \rangle.
\]

\paragraph{IFGT error analysis.}

The error incurred by truncating the sum at degree $\tau-1$ is given by 
\[
\textsf{Err}(\tau) 
= 
\big| \sum_{p \in P} \sum_{q \in Q} K(p,q) \mu(p) \nu(q) - \sum_{p \in P} \sum_{q \in Q} \IP{\tilde \phi(p)}{\tilde \phi(q)} \big| 
\leq
\sum_{p \in P} \sum_{q \in Q} \mu(p) \nu(q) \frac{2^\tau}{\tau!} \Delta^{2\tau}
=
W^2 \frac{2^\tau}{\tau!} \Delta^{2 \tau}.
\]

Set $\eps W^2= \textsf{Err}(\tau)$. Applying Stirling's approximation, we solve for $\tau$ in $\log (1/\eps) \geq \tau \log (\tau/4 \Delta^2)$. This yields the bounds $\tau = O(\Delta^2)$ and $\tau = O(\log (1/\eps))$.  Thus our error bound holds for $\tau = O(\Delta^2 + \log (1/\eps))$.  Using $\rho = O(\tau^d)$, we obtain the following result.

\begin{lemma} \label{lem:IFGT-feat}
There exists a mapping $\tilde \phi : \c{P} \cup \c{Q} \to \reals^\rho$ with $\rho = O(\Delta^{2d} + \log^d (1/\eps))$ so
\[
\bigg|\sum_{p \in \c{P}} \sum_{q \in \c{Q}} K(p,q) \mu(p) \nu(q) - \sum_{p \in \c{P}} \sum_{q  \in \c{Q}} \IP{\tilde \phi(p)}{\tilde \phi(q)} \bigg| \leq \eps W^2.
\]
\end{lemma}

\subsection{Summary of Feature Maps}
We have developed three different bounds on the dimension required for feature maps that approximate $\kappa(\c{P}, \c{Q})$ to within $\eps W^2$.  
\begin{itemize} \denselist
\item[\textsf{IFGT}:] $\rho = O(\Delta^{2d} + \log^d(1/\eps))$.  Lemma \ref{lem:IFGT-feat}.
Advantages: deterministic, independent of $n$, logarithmic dependence on $1/\eps$.  
Disadvantages:  polynomial dependence on $\Delta$, exponential dependence on $d$.  
\item[\textsf{Random-points}:] $\rho = O((1/\eps^2) \log(n/\delta))$.  Lemma \ref{lem:rand-feat}.
Advantages: independent of $\Delta$ and $d$.  
Disadvantages: randomized, dependent on $n$, polynomial dependence on $1/\eps$.
\item[\textsf{Random-domain}:] $\rho = O((d/\eps^2) \log(\Delta/\eps \delta))$.  (above)
Advantages: independent of $n$, logarithmic dependence on $\Delta$, polynomial dependence on $d$.  
Disadvantages:  randomized, dependence on $\Delta$ and $d$, polynomial dependence on $1/\eps$.  
\end{itemize}

For simplicity, we (mainly) use the \textsf{Random-points} based result from Lemma \ref{lem:rand-feat} in what follows. If appropriate in a particular application, the other bounds may be employed.

\paragraph{Feature-based computation of $D_K$.}
\label{sec:fast-current-norm}

As before, we can decompose $
D_K^2(\c{P},\c{Q}) = \kappa(\c{P},\c{P}) + \kappa(\c{Q},\c{Q}) - 2\kappa(\c{P},\c{Q})$
and use Lemma \ref{lem:rand-feat} to approximate each of $\kappa(\c{P},\c{P}), \kappa(\c{Q},\c{Q})$, and $\kappa(\c{P},\c{Q})$ with error $\eps W^2/4$.  

\begin{theorem} \label{thm:fast-CN}
We can compute a value $U$ in time 
$O((n/\eps^2) \log(n/\delta))$ such that $|U - D_K^2(\c{P},\c{Q})| \leq \eps W^2$, with probability at least $1-\delta$. 
\end{theorem}

\paragraph{A nearest-neighbor algorithm.}
The feature map does more than yield efficient algorithms for the kernel distance. As a representation for shapes and distributions, it allows us to solve other data analysis problems on shape spaces using off-the-shelf methods that apply to points in Euclidean space. As a simple example of this, we can combine the \textsf{Random-points} feature map with known results on approximate nearest-neighbor search in Euclidean space~\cite{DBLP:conf/focs/AndoniI06} to obtain the following result.

\begin{lemma}
 Given a collection of $m$ point sets $\c{C} = \{\c{P}_1, \c{P}_2, \ldots, \c{P}_m\}$, and a query surface $\c{Q}$, we can compute the $c$-approximate nearest neighbor to $\c{Q}$ in $\c{C}$ under the kernel distance in time $O(\rho m^{1/c^{2}+o(1)})$ query time using $O(\rho m^{1+1/c^{2}+o(1)})$ space and preprocessing.
\end{lemma}

\section{Coresets for the Kernel Distance}
\label{sec:core-set-current}

The kernel norm (and distance) can be approximated in near-linear time; however, this may be excessive for large data sets.  Rather, we extract a small subset (a coreset) $\c{S}$ from the input $\c{P}$ such that the kernel distance between $\c{S}$ and $\c{P}$ is small. By triangle inequality, $\c{S}$ can be used as a proxy for $\c{P}$. Specifically, we extend the notion of $\eps$-samples for range spaces to handle non-binary range spaces defined by kernels.

\paragraph{Background on range spaces.}
Let $\xi(P)$ denote the total weight of a set of points $P$, or cardinality if no weights are assigned.
Let $P \subset \b{R}^d$ be a set of points and let $\c{A}$ be a family of subsets of $P$.   For examples of $\c{A}$, let $\c{B}$ denote the set of all subsets defined by containment in a ball and let $\c{E}$ denote the set of all subsets defined by containment in ellipses.  
We say $(P, \c{A})$ is a \emph{range space}.  Let $\bar \xi_P(A) = \xi(A)/\xi(P)$.  
An \emph{$\eps$-sample} (or \emph{$\eps$-approximation}) of $(P, \c{A})$ is a subset $Q \subset P$ such that 
\[
\max_{A \in \c{A}} \left| \bar \xi_Q(Q \cap A) - \bar \xi_P(P \cap A)\right| \leq \eps.
\]


To create a coreset for the kernel norm, we want to generalize these notions of $\eps$-samples to non-binary ($(0,1)$-valued instead of $\{0,1\}$-valued) functions, specifically to kernels.  
For two point sets $P,Q$, define $\bar \kappa(P,Q) = (1/\xi(P))(1/\xi(Q)) \sum_{p \in P} \sum_{q \in Q} K(p,q)$, and when we have a singleton set $Q = \{q\}$ and a subset $P' \subseteq P$ then we write $\bar \kappa_P(P',q) = (1/\xi(P)) \sum_{p \in P'} K(p,q)$.  Let $K^+ = \max_{p,q \in P} K(p,q)$ be the maximum value a kernel can take on a dataset $P$, which can be normalized to $K^+=1$.  
We say a subset of $S \subset P$ is an \emph{$\eps$-sample of $(P,K)$} if 
\[
\max_q \left| \bar \kappa_P(P,q) - \bar \kappa_S(S,q) \right| \leq \eps K^+.
\]

The standard notion of VC-dimenion~\cite{VC71} (and related notion of shattering dimension) is fundamentally tied to the binary ($\{0,1\}$-valued) nature of ranges, and as such, it does not directly apply to $\eps$-samples of $(P,K)$.  
Other researchers have defined different combinatorial dimensions that can be applied to kernels~\cite{DGL96,KS94,ABCH97,Vap89}.  The best result is based on $\gamma$-fat shattering dimension $\textsc{f}_\gamma$~\cite{KS94}, defined for a family of $(0,1)$-valued functions $\c{F}$ and a ground set $P$.  A set $Y \subset P$ is \emph{$\gamma$-fat shattered} by $\c{F}$ if there exists a function $\alpha : Y \to [0,1]$ such that for all subsets $Z \subseteq Y$ there exists some $F_Z \in \c{F}$ such that for every $x \in Z$ $F_Z(x) \geq \alpha(x) + \gamma$ and for every $x \in Y \setminus Z$ $F_Z(x) \leq \alpha(x) - \gamma$.  Then $\textsc{f}_\gamma = \xi(Y)$ for the largest cardinality set $Y \subset P$ that can be $\gamma$-fat shattered.  
Bartlett \etal~\cite{BLW96} show that a random sample of $O((1/\eps^2) (\textsc{f}_\gamma \log^2(\textsc{f}_\gamma/\eps) + \log(1/\delta))$ elements creates an $\eps$-sample (with probability at least $1-\delta$) with respect to $(P,\c{F})$ for $\gamma = \Omega(\eps)$.  
Note that the $\gamma$-fat shattering dimension of Gaussian and other symmetric kernels in $\b{R}^d$ is $d+1$ (by setting $\alpha(x) = .5$ for all $x$), the same as balls $\c{B}$ in $\b{R}^d$, so this gives a random-sample construction for $\eps$-samples of $(P,K)$ of size $O((d/\eps^2) (\log^2(1/\eps) + \log(1/\delta))$.


In this paper, we improve this result in two ways by directly relating a kernel range space $(P,K)$ to a similar (binary) range space $(P,\c{A})$.  
First, this improves the random-sample bound because it uses sample-complexity results for binary range spaces that have been heavily optimized.  
Second, this allows for all deterministic $\eps$-sample constructions 
 (which have no probability of failure) and can have much smaller size.  


\paragraph{Constructions for $\eps$-samples.}

Vapnik and Chervonenkis~\cite{VC71} showed that the complexity of $\eps$-samples is tied to the VC-dimension of the range space.  That is, given a range space $(X, \c{A})$ a subset $Y \subset X$ is said to be \emph{shattered} by $\c{A}$ if all subsets of $Z \subset Y$ can be realized as $Z = Y \cap R$ for $R \in \c{A}$.  Then the \emph{VC-dimension} of a range space $(X,\c{A})$ is the cardinality of the largest subset $Y \subset X$ that can be shattered by $\c{A}$.  
Vapnik and Chervonenkis~\cite{VC71} showed that if the VC-dimension of a range space $(X,\c{A})$ is $\nu$, then a random sample $Y$ of $O((1/\eps^2)(\nu \log (1/\eps) + \log(1/\delta))$ points from $X$ is an $\eps$-sample with probability at least $1-\delta$.  
This bound was improved to $O((1/\eps^2) (\nu + \log 1/\delta))$ by Talagrand~\cite{Tal94,LLS01}.  

Alternatively, Matousek~\cite{Mat91} showed that $\eps$-samples of size $O((\nu/\eps^2) \log (\nu/\eps))$ could be constructed deterministically, that is there is no probability of failure.  A simpler algorithm with more thorough runtime analysis is presented in Chazelle and Matousek~\cite{CM96}, which runs in $O(d)^{3d} |X|(1/\eps)^{2\nu}\log^{\nu}(1/\eps)$ time.  
Smaller $\eps$-samples exist; in particular Matousek, Welzl, and Wernisch~\cite{MWW93} and improved by Matousek~\cite{Mat95} show that $\eps$-samples exist of size $O((1/\eps)^{2-2/(\nu+1)})$, based on a discrepancy result that says there exists a labeling $\chi : X \to \{-1,+1\}$ such that $\max_{R \in \c{A}} \sum_{x \in R \cap X} \chi(x) \leq O(|X|^{1/2 - 1/2\nu} \log^{1/2} |X|)$.  
It is alluded by Chazelle~\cite{Cha00} that if an efficient construction for such a labeling existed, then an algorithm for creating $\eps$-samples of size $O((1/\eps)^{2-2/(\nu+1)}\log(\nu/\eps)^{2-1/{d+1}})$ would follow, see also Phillips~\cite{Phi08}.  
Recently, Bansal~\cite{Ban10} provided a randomized polynomial algorithm for the entropy method, which is central in proving these existential coloring bounds.  This leads to an algorithm that runs in time $O(|X| \cdot \poly{1/\eps})$, as claimed by Charikar \etal~\cite{CNN11,N11}.

An alternate approach is through the VC-dimension of the \emph{dual range space} $(\c{A},\bar{\c{A}})$, of (primal) range space $(X,\c{A})$, where $\bar{\c{A}}$ is the set of subsets of ranges $\c{A}$ defined by containing the same element of $X$.  
In our context, for range spaces defined by balls $(\b{R}^d,\c{B})$ and ellipses of fixed orientation $(\b{R}^d, \c{E})$ their dual range spaces have VC-dimension $\bar \nu = d$.  
Matousek~\cite{Mat99} shows that a technique of matching with low-crossing number~\cite{CW89} along with Haussler's packing lemma~\cite{Hau95} can be used to construct a low discrepancy coloring for range spaces where $\bar \nu$, the VC-dimension of the dual range space is bounded.  This technique can be made deterministic and runs in $\poly{|X|}$ time.  Invoking this technique in the Chazelle and Matousek~\cite{CM96} framework yields an $\eps$-sample of size $O((1/\eps)^{2-2/(\bar \nu +1)} (\log(1/\eps))^{2-1/(\bar \nu+1)})$ in $O(|X| \cdot \poly{1/\eps})$ time. 
Specifically, we attain the following result:
\begin{lemma}
\label{lem:e-samp-ball}
For discrete ranges spaces $(X,\c{B})$ and $(X,\c{E})$ for $X \in \b{R}^d$ of size $n$, we can construct an $\eps$-sample of size $O((1/\eps)^{2-2/(d+1)}(\log (1/\eps))^{2-1/d+1})$ in $O(n \cdot \emph{\poly{1/\eps}})$ time.  
\end{lemma}

For specific range spaces, the size of $\eps$-samples can be improved beyond the VC-dimension bound.  Phillips~\cite{Phi08} showed for ranges $\c{R}_d$ consisting of axis-aligned boxes in $\b{R}^d$, that $\eps$-samples can be created of size $O((1/\eps) \cdot \log^{2d} (1/\eps))$.  This can be generalized to ranges defined by $k$ predefined normal directions of size $O((1/\eps) \cdot \log^{2k}(1/\eps))$.  These algorithms run in time $O(|X| (1/\eps^3) \textrm{poly}\log (1/\eps))$.
And for intervals $\c{I}$ over $\b{R}$, $\eps$-samples of $(X, \c{I})$ can be created of size $O(1/\eps)$ by sorting points and retaining every $\eps|X|$th point in the sorted order~\cite{LP09}.

\paragraph{$\eps$-Samples for kernels.}
The \emph{super-level set} of a kernel given one input $q \in \b{R}^d$ and a value $v \in \b{R}^+$, is the set of all points $p \in \b{R}^d$ such that $K(p,q) \geq v$.  We say that a kernel is \emph{linked} to a range space $(\b{R}^d, \c{A})$ if for every possible input point $q \in \b{R}^d$ and any value $v \in \b{R}^+$ that the super-level set of $K(\cdot,q)$ defined by $v$ is equal to some $H \in \c{A}$.  For instance multi-variate Gaussian kernels with no skew are linked to $(\b{R}^d, \c{B})$ since all super-level sets are balls, and multi-variate Gaussian kernels with non-trivial covariance are linked to $(\b{R}^d, \c{E})$ since all super-level sets are ellipses.

\begin{theorem}
For any kernel $K : \c{M} \times \c{M} \to \b{R}^+$ linked to a range space $(\c{M}, \c{A})$, an $\eps$-sample $S$ of $(P, \c{A})$ for $S \subseteq \c{M}$ is a $\eps$-sample of $(P,K)$.  
\label{thm:kernel-sample}
\end{theorem}

\emph{A (flawed) attempt at a proof may proceed by considering a series of approximate level-sets, within which each point has about the same function value.  Since $S$ is an $\eps$-sample of $(P,\c{A})$, we can guarantee the density of $S$ and $P$ in each level set is off by at most $2\eps$.  However, the sum of absolute error over all approximate level-sets is approximately $\eps K^+$ times the number of approximate level sets.  This analysis fails because it allows error to accumulate; however, a more careful application of the $\eps$-sample property shows it cannot.  A correct analysis follows using a charging scheme which prevents the error from accumulating.  }

\begin{proof}
We can sort all $p_i \in P$ in similarity to $q$ so that $p_i < p_j$ (and by notation $i<j$) if $K(p_i,q) > K(p_j,q)$.  Thus any super-level set containing $p_j$ also contains $p_i$ for $i<j$.  We can now consider the one-dimensional problem on this sorted order from $q$.  

We now count the deviation $E(P,S,q) = \bar \kappa_P(P,q) - \bar \kappa_S(S,q)$ from $p_1$ to $p_n$ using a charging scheme.  That is each element $s_j \in S$ is charged to $\xi(P)/\xi(S)$ points in $P$.  For simplicity we will assume that $k = \xi(P)/\xi(S)$ is an integer, otherwise we can allow fractional charges.  
We now construct a partition of $P$ slightly differently, for positive and negative $E(P,S,q)$ values, corresponding to undercounts and overcounts, respectively.  

\textbf{Undercount of $\bar \kappa_S(S,q)$.}
For undercounts, we partition $P$ into $2\xi(S)$ (possibly empty) sets $\{P'_1, P_1, P'_2,$ $P_2, \ldots, P'_{\xi(S)}, P_{\xi(S)}\}$ of consecutive points by the sorted order from $q$.  
Starting with $p_1$ (the closest point to $q$) we place points in sets $P'_j$ or $P_j$ following their sorted order.  Recursively on $j$ and $i$, starting at $j=1$ and $i=1$, we place each $p_i$ in $P'_j$ as long as $K(p_i,q) > K(s_j,q)$ (this may be empty).  Then we place the next $k$ points $p_i$ into $P_j$.  After $k$ points are placed in $P_j$, we begin with $P'_{j+1}$, until all of $P$ has been placed in some set.  Let $t \leq \xi(S)$ be the index of the last set $P_j$ such that $\xi(P_j) = k$.  
Note that for all $p_i \in P_j$ (for $j \leq t$) we have $K(s_j,q) \geq K(p_i,q)$, thus $\bar \kappa_S(\{s_j\},q) \geq \bar \kappa_P(P_j,q)$. 

We can now bound the undercount as 
\[
E(P,S,q) = 
\sum_{j=1}^{\xi(S)} \left(\bar \kappa_P(P_j,q) - \bar \kappa_S(\{s_j\},q) \right)
+
\sum_{j=1}^{\xi(S)} \bar \kappa_P(P'_j,q)
 \leq 
\sum_{j=1}^{t+1} \bar \kappa_P(P'_j,q)
\]
since the first term is at most $0$ and since $\xi(P'_j) = 0$ for $j > t+1$.  
Now consider a super-level set $H \in \c{A}$ containing all points before $s_{t+1}$;  $H$ is the smallest range that contains every non-empty $P'_j$.  
Because (for $j \leq t$) each set $P_j$ can be charged to $s_j$, then $\sum_{j=1}^t \xi(P_j \cap H) = k \cdot \xi(S \cap H)$.  
And because $S$ is an $\eps$-sample of $(P,\c{A})$, then 
$\sum_{j=1}^{t+1} \xi(P'_j) = \left( \sum_{j=1}^{t+1} \xi(P'_j) + \sum_{j=1}^t \xi(P_j \cap H)\right) - k \cdot \xi(S \cap H) \leq \eps \xi(P)$.   
We can now bound
\[
E(P,S,q) 
\leq 
\sum_{j=1}^{t+1} \bar \kappa_P(P_j',q)
=
\sum_{j=1}^{t+1} \sum_{p \in P_j'} \frac{K(p,q)}{\xi(P)}
\leq 
\frac{1}{\xi(P)} \sum_{j=1}^{t+1} \xi(P_j') K^+
\leq
\frac{1}{\xi(P)} (\eps \xi(P)) K^+
=
\eps K^+.
\]

\textbf{Overcount of $\bar \kappa_S(S,q)$:}
The analysis for overcounts is similar to undercounts, but we construct the partition in reverse and the leftover after the charging is not quite as clean to analyze.  
For overcounts, we partition $P$ into $2\xi(S)$ (possibly empty) sets $\{P_1, P'_1, P_2, P'_2, \ldots, P_{\xi(S)},$ $P'_{\xi(S)}\}$ of consecutive points by the sorted order from $q$.  
Starting with $p_n$ (the furthest point from $q$) we place points in sets $P'_j$ or $P_j$ following their reverse-sorted order.  Recursively on $j$ and $i$, starting at $j=\xi(S)$ and $i=n$, we place each $p_i$ in $P'_j$ as long as $K(p_i,q) < K(s_j,q)$ (this may be empty).  Then we place the next $k$ points $p_i$ into $P_j$.  After $k$ points are placed in $P_j$, we begin with $P'_{j-1}$, until all of $P$ has been placed in some set.  
Let $t \leq \xi(S)$ be the index of the last set $P_j$ such that $\xi(P_j) = k$ (the smallest such $j$).  
Note that for all $p_i \in P_j$ (for $j \geq t$) we have $K(s_j,q) \leq K(p_i,q)$, thus $\bar \kappa_S(\{s_j\},q) \leq \bar \kappa_P(P_j,q)$. 

We can now bound the (negative) undercount as 
\begin{align*}
E(P,S,q) = &
\sum_{j=\xi(S)}^t \left(\bar \kappa_P(P_j,q) - \bar \kappa_S(\{s_j\},q) \right)
+
\sum_{j=t-1}^1 \left(\bar \kappa_P(P_j,q) - \bar \kappa_S(\{s_j\},q) \right)
+
\sum_{j=1}^{\xi(S)} \bar \kappa_P(P'_j,q)
\\ \geq &
\left(\bar \kappa_P(P_{t-1},q) - \bar \kappa_S(\{s_{t-1}\}, q) \right) - \sum_{j=t-2}^1 \bar \kappa_S(\{s_j\},q),
\end{align*}
since the first full term is at least $0$, as is each $\bar \kappa_P(P_j,q)$ and $\bar \kappa_P(P'_j,q)$ term in the second and third terms.  We will need the one term $\bar \kappa_P(P_{t-1},q)$ related to $P$ in the case when $1 \leq \xi(P_{t-1}) < k$.  

Now, using that $S$ is an $\eps$-sample of $(P,\c{A})$, we will derive a bound on $t$, and more importantly $(t-2)$.   We consider the maximal super-level set $H \in \c{A}$ such that no points $H \cap P$ are in $P'_j$ for any $j$.  This is the largest set where each point $p \in P$ can be charged to a point $s \in S$ such that $K(p,q) > K(s,q)$, and thus presents the smallest (negative) undercount.  
In this case, $H \cap P = \cup_{j=1}^s P_j$ for some $s$ and $H \cap S = \cup_{j=1}^s \{s_j\}$.  Since $t \leq s$, then $\xi(H \cap P) = (s-t+1) k +\xi(P_{t-1})= (s-t+1) \xi(P)/\xi(S) + \xi(P_{t-1})$ and $\xi(H \cap S) = s$.  Thus
\[
\eps 
\geq
\bar \xi_S(H \cap S) - \bar \xi_P(H \cap P)
=
\frac{s}{\xi(S)} - \frac{(s-t+1) \xi(P)/\xi(S)}{\xi(P)} - \frac{\xi(P_{t-1})}{\xi(P)}
\geq
\frac{t-1}{\xi(S)} - \frac{\xi(P_{t-1})}{\xi(P)}.
\]
Thus $(t -2) \leq \eps \xi(S)+\xi(P_{t-1}) (\xi(S)/\xi(P)) - 1$.  Letting $p_i = \min_{i' \in P_{t-1}} K(p_{i'},q)$ (note $K(p_i,q) \geq K(s_{t-1},q)$)
\begin{align*}
E(P,S,q) 
&\geq 
\frac{\kappa(P_{t-1},q)}{\xi(P)} - \frac{K(s_{t-1},q)}{\xi(S)} - \left(\eps \xi(S)+ \xi(P_{t-1}) \frac{\xi(S)}{\xi(P)} - 1\right) \frac{K^+}{\xi(S)} 
\\ &= 
- \eps K^+ + K^+ \left(\frac{k - \xi(P_{t-1})}{\xi(P)}\right) - \frac{k \cdot K(s_{t-1},q) - \kappa(P_{t-1},q) }{\xi(P)}
\\ &\geq
- \eps K^+ + K^+ \left(\frac{k - \xi(P_{t-1})}{\xi(P)}\right) - K(p_i,q) \left(\frac{k  - \xi(P_{t-1})}{\xi(P)}\right)
\geq
-\eps K^+.  \qedhere
\end{align*}
\end{proof}

\begin{corollary}\label{cor:Gaussian-coreset}
For a Gaussian kernel, any $\eps$-sample $S$ for $(P, \c{B})$  (or for $(P, \c{E})$ if we consider covariance) guarantees that for any query $q \in \b{R}^d$ that 
$
\left| \bar\kappa_P(P,q) - \bar\kappa_S(S,q) \right| \leq \eps K^+. 
$
\end{corollary}

\paragraph{Coreset-based computation of kernel distance.}

For convenience here we assume that our kernel has been normalized so $K^+ = 1$.  
Let $P$ be an $\eps$-sample of $(\c{P}, K)$, and all points $p \in P$ have uniform weights so $\xi(P) = \xi(\c{P}) = W$.  
Then for any $q \in \b{R}^d$ 
\[
\eps 
\geq 
| \bar \kappa_P(P,q) - \bar \kappa_{\c{P}}(\c{P},q) |  
= 
\left| \frac{\kappa(P,q)}{\xi(P)} - \frac{\kappa(\c{P},q)}{\xi(\c{P})} \right|.
\] 
and hence 
\[
\left| \kappa(P,q) - \kappa(\c{P},q)\right|
\leq
\eps \xi(\c{P}) = \eps W.
\]
It follows that if $Q$ is also an $\eps$-sample for $(\c{Q},K)$, then 
$\| \kappa(P,Q) - \kappa(\c{P},\c{Q}) \| \leq 2\eps W^2.$
Hence, via known constructions of $\eps$-samples randomized~\cite{VC71,Tal94} or deterministic~\cite{Mat91,CM96,MWW93,Mat95,STZ04,Phi08} (which can be applied to weighted point sets to create unweighted ones~\cite{Mat91}) and Theorem \ref{thm:kernel-sample} we have the following theorems.  The first shows how to construct a coreset with respect to $D_K$.  

\begin{theorem}\label{thm:random-coreset}
Consider any kernel $K$ linked with $(\b{R}^d,\c{B})$ and objects $\c{P},\c{Q} \subset \c{M} \subset \b{R}^d$, each with total weight $W$, and for constant $d$.  We can construct sets $P \subset \c{P}$ and $Q \subset \c{Q}$ such that $|D_K(\c{P},\c{Q}) - D_K(P,Q)| \leq \eps W^2$ of size:
\begin{itemize} \denselist
\item $O((1/\eps^{2 - 1/(d+1)}) \log^{2-1/d+1}(1/\eps))$, via Lemma \ref{lem:e-samp-ball}; or 
\item $O(1/\eps^2)(d+\log(1/\delta)))$ via random sampling (correct with probability at least $(1-\delta)$).
\end{itemize}
%
\end{theorem}

We present an alternative sampling condition to Theorem \ref{thm:random-coreset} in Appendix \ref{sec:coreset}.  It has larger dependence on $\eps$, and also has either dependence on $\Delta$ or on $\log n$ (and is independent of $K^+$).  
Also in Appendix \ref{sec:coreset} we show that it is NP-hard to optimize $\eps$ with a fixed subset size $k$. 

The associated runtimes are captured in the following algorithmic theorem. 

\begin{theorem}\label{thm:aprox-KD}
When $K$ is linked to $(\b{R}^d,\c{B})$, 
we can compute a number $\tilde D$ such that $|D_K(\c{P},\c{Q}) - \tilde D| \leq \eps$ in time:
\begin{itemize} \denselist
\item $O(n \cdot (1/\eps^{2d+2}) \log^{d+1}(1/\eps))$; or
\item $O(n + (\log n) \cdot ((1/\eps^2) \log(1/\delta)) + (1/\eps^4)\log^2(1/\delta))$ that is correct with probability at least $1-\delta$.
\end{itemize}
\end{theorem}

Notice that these results automatically work for any kernel linked with $(\b{R}^d, \c{B})$ (or more generally with $(\b{R}^d, \c{E})$) with no extra work;  this includes not only Gaussians (with non-trivial covariance), but any other standard kernel such as triangle, ball or Epanechnikov.


\section{Minimizing the Kernel Distance under Translation and Rotation}
\label{sec:an-fptas-minimizing}

We attack the problem of minimizing the kernel distance between $\c{P}$ and $\c{Q}$ under a set of transformations: translations or translations and rotations.  
A \emph{translation} $T \in \b{R}^d$ is a vector so $\c{Q} \oplus T = \{q+T \mid q \in Q\}$.  The translation 
$
T^* = \arg \min_{T \in \b{R}^d} D_K(\c{P}, \c{Q} \oplus T),
$
applied to $\c{Q}$ minimizes the kernel norm.  
A \emph{rotation} $R \in \SO{d}$ can be represented as a special orthogonal matrix.  
We can write $R \circ \c{Q} = \{R(q) \mid q \in Q\}$, where $R(q)$ rotates $q$ about the origin, preserving its norm.  The set of a translation and rotation
$
(T^\star, R^\star) = \arg \min_{(T,R) \in \b{R}^d \times \SO{d}} D_K(\c{P}, R \circ (\c{Q} \oplus T))
$
applied to $\c{Q}$ minimizes the kernel norm.  

\paragraph{Decomposition.}
In minimizing $D_K(\c{P}, R \circ (\c{Q} \oplus T))$ under all translations and rotations, we can reduce this to a simpler problem.  The first term $\kappa(\c{P}, \c{P}) = \sum_{p_1 \in \c{P}} \sum_{p_2 \in \c{P}} \mu(p_1) \mu(p_2) K(p_1, p_2)$ has no dependence on $T$ or $R$, so it can be ignored.  And the second term $\kappa(\c{Q}, \c{Q}) = \sum_{q_1 \in \c{Q}} \sum_{q_2 \in \c{Q}} \nu(q_1) \nu(q_2) K(R(q_1 + T), R(q_2+T))$ can also be ignored because it is invariant under the choice of $T$ and $R$.  Each subterm $K(R(q_1 +T), R(q_2 + T))$ only depends on $||R(q_1 + T) - R(q_2 + T)|| = ||q_1 - q_2||$, which is also independent of $T$ and $R$.  Thus we can rephrase the objective as finding
\\
$\displaystyle{\hspace{0.5in}
(T^\star,R^\star) = \arg \max_{(T,R) \in \b{R}^d \times \SO{d}} \sum_{p \in \c{P}} \sum_{q \in \c{Q}} \mu(p) \nu(q) K(p, R(q + T)) = \arg \max_{(T,R) \in \b{R}^d \times \SO{d}} \kappa(\c{P},R\circ (\c{Q} \oplus T)).
}$

We start by providing an approximately optimal translation.  Then we adapt this algorithm to handle both translations and rotations.

\subsection{Approximately Optimal Translations}
We describe, for any parameter $\eps > 0$, an algorithm for a translation $\hat{T}$ such that $D_K^2(\c{P}, \c{Q} \oplus \hat{T}) - D_K^2(\c{P}, \c{Q} \oplus T^*) \leq \eps W^2$.
We begin with a key lemma providing analytic structure to our problem.

\begin{lemma}
$\kappa(\c{P}, \c{Q} \oplus T^*) \geq W^2/n^2$.
\label{lem:lb1}
\end{lemma}
\begin{proof} 
When $T \in \b{R}^d$ aligns $q \in Q$ so $q + T = p$ for $p \in P$ it ensures that $K(p,q) = 1$.  We can choose the points $p$ and $q$ such that $\mu(p)$ and $\nu(q)$ are as large as possible.  They must each be at least $W/n$, so $K(p,q) \mu(p) \nu(q) \geq W^2/n^2$.  All other subterms in $\kappa(\c{P},\c{Q}\oplus T)$ are at least $0$.   Thus $\kappa(\c{P},\c{Q} \oplus T) \geq W^2/n^2$.  
\end{proof}

Thus, if $\kappa(\c{P},\c{Q} \oplus T^*) \geq W^2/n^2$, then for some pair of points $p \in \c{P}$ and $q \in \c{Q}$ we must have $\mu(p) \nu(q) K(p,q+T^*) \geq \mu(p)\nu(q)/n^2$, i.e. $K(p,q+T^*) \geq 1/n^2$.  Otherwise, if all $n^2$ pairs $(p,q)$ satisfy $\mu(p)\nu(q)K(p,q+T^*) < \mu(p)\nu(q)/n^2$, then 
\[
\kappa(\c{P},\c{Q} \oplus T^*) 
= 
\sum_{p \in \c{P}}\sum_{q \in \c{Q}} \mu(p)\nu(q) K(p,q+T^*) 
< 
\sum_{p \in \c{P}}\sum_{q \in \c{Q}} \mu(p)\nu(q) /n^2
=
W^2/n^2.
\]  
Thus some pair $p \in \c{P}$ and $q \in \c{Q}$ must satisfy $||p-(q+T^*)|| \leq \sqrt{\ln (n^2)}$, via Lemma \ref{lem:G-dist} with $\gamma = 1/(n^2)$.

Let $G_\eps$ be a grid on $\b{R}^d$ so that when any point $p \in \b{R}^d$ is snapped to the nearest grid point $g \in G_\eps$, we guarantee that $||g-p|| \leq \eps$.  
We can define an orthogonal grid $G_\eps = \{(\eps/\sqrt{d}) z \mid z \in \b{Z}^d \}$, where $\b{Z}^d$ is the $d$-dimensional lattice of integers.  
Let $\c{G}[\eps,p,\Lambda]$ represent the subset of the grid $G_\eps$ that is within a distance $\Lambda$ of the point $p$.  
In other words, $\c{G}[\eps, p, \Lambda] = \{g \in G_\eps \mid ||g - p|| \leq \Lambda\}$.  

\paragraph{Algorithm.}
These results imply the following algorithm.  
For each point $p \in \c{P}$, for each $q \in \c{Q}$, and for each $g \in \c{G}[\eps/2, p, \sqrt{\ln(n^2)}]$ we consider the translation $T_{p,q,g}$ such that $q + T_{p,q,g} = g$.  We return the translation $T_{p,q,g}$ which maximizes $\kappa(\c{P}, \c{Q} \oplus T_{p,q,g})$, by evaluating $\kappa$ at each such translation of $\c{Q}$.  

\begin{theorem}
The above algorithm runs in time $O((1/\eps)^d n^{4} \log^{d/2} n)$, 
for fixed $d$, and is guaranteed to find a translation $\hat{T}$
such that $D_K^2(\c{P}, \c{Q} \oplus \hat{T}) - D_K^2(\c{P}, \c{Q} \oplus T^*) \leq \eps W^2$.  
\label{lem:trans}
\end{theorem}
\begin{proof}
We know that the optimal translation $T^*$ must result in some pair of points $p \in \c{P}$ and $q \in \c{Q}$ such that $||p - (q+T^*)|| \leq \sqrt{\ln (n^2)}$ by Lemma \ref{lem:lb1}.  So checking all pairs $p \in \c{P}$ and $q \in \c{Q}$, one must have $||p - q|| \leq \sqrt{\ln(n^2)}$.  
Assuming we have found this closest pair, $p \in \c{P}$ and $q \in \c{Q}$, we only need to search in the neighborhood of translations $T = p - q$.  

Furthermore, for some translation $T_{p,q,g} = g-q$ we can claim that $\kappa(\c{P},\c{Q} \oplus T^*) - \kappa(\c{P},\c{Q} \oplus T_{p,q,g}) \leq \eps$.  Since $||T^* - T_{p,q,g}|| \leq \eps/2$, we have the bound on subterm $| K(p,q+T^*) - K(p,q+T_{p,q,g})| \leq \eps/2$, by Lemma \ref{lem:grid-eps}.  In fact, for every other pair $p' \in \c{P}$ and $q' \in \c{Q}$, we also know $| K(p',q'+T^*) - K(p',q'+T_{p,q,g})| \leq \eps/2$.  Thus the sum of these subterms has error at most $(\eps/2) \sum_{p \in \c{P}} \sum_{q \in \c{Q}} \mu(p) \nu(q) = (\eps/2)W^2$.  

Since, the first two terms of $D_K^2(\c{P},\c{Q} \oplus T)$ are unaffected by the choice of $T$, this provides an $\eps$-approximation for $D_K^2(\c{P},\c{Q} \oplus T)$ because all error is in the $(-2)\kappa(\c{P},\c{Q} \oplus T)$ term.  

For the runtime we need to bound the number of pairs from $\c{P}$ and $\c{Q}$ (i.e. $O(n^2)$), the time to calculate $\kappa(\c{P},\c{Q} \oplus T)$ (i.e. $O(n^2)$), and finally the number of grid points in $\c{G}[\eps/2, p, \sqrt{\ln (n^2)}]$.  The last term requires $O((1/\eps)^d)$ points per unit volume, and a ball of radius $\sqrt{\ln (n^2)}$ has volume $O(\log^{d/2} n)$, resulting in $O((1/\eps)^d \log^{d/2} n)$ points.  This product produces a total runtime of $O((1/\eps)^d n^4 \log^{d/2} n)$.  
\end{proof}

For a constant dimension $d$, using Theorem \ref{thm:random-coreset} to construct a coreset, we can first set $n = O((1/\eps^2)\log (1/\delta))$ and now 
the time to calculate $\kappa(\c{P},\c{Q} \oplus T)$ is $O((1/\eps^4) \log^2 (1/\delta))$ after spending $O(n + (1/\eps^2) \log (1/\delta) \log n)$ time to construct the coresets. 
Hence the total runtime is 
\[
O(n + \log n (1/\eps^2)( \log(1/\delta)) + (1/\eps^{d+8}) \cdot \log^{d/2} ((1/\eps) \log(1/\delta))\log^4(1/\delta)),
\]
and is correct with probability at least $1-\delta$.  

\begin{theorem}
For fixed $d$, in 
\[
O(n + \log n (1/\eps^2)( \log(1/\delta)) + (1/\eps^{d+8}) \cdot \log^{d/2} ((1/\eps) \log(1/\delta))\log^4(1/\delta))
\]
time we can find a translation $\hat T$ such that 
$D_K^2(\c{P}, \c{Q} \oplus \hat{T}) - D_K^2(\c{P}, \c{Q} \oplus T^*) \leq \eps W^2$, 
with probability at least $1-\delta$.  
\label{thm:trans}
\end{theorem}



\subsection{Approximately Optimal Translations and Rotations}

For any parameter $\eps > 0$, we describe an algorithm to find a translation $\hat T$ and a rotation $\hat R$ such that 
\[
D_K^2(\c{P}, \hat R \circ (\c{Q} \oplus \hat T)) - D_K^2(\c{P}, R^\star \circ (\c{Q} \oplus T^\star)) \leq \eps W^2.
\]
We first find a translation to align a pair of points $p \in \c{P}$ and $q \in \c{Q}$ within some tolerance (using a method similar to above, and using Lemma~\ref{lem:G-dist} to ignore far-away pairs) and then rotate $\c{Q}$ around $q$.  This deviates from our restriction above that $\hat R \in \SO d$ rotates about the origin, but can be easily overcome by performing the same rotation about the origin, and then translating $\c{Q}$ again so $q$ is at the desired location.  Thus, after choosing a $q \in \c{Q}$ (we will in turn choose each $q' \in \c{Q}$) we let all rotations be about $q$ and ignore the extra modifications needed to $\hat T$ and $\hat R$ to ensure $\hat R$ is about the origin.  


Given a subset $S \subset \c{Q}$ of fewer than $d$ points and a pair $(p,q) \in \c{P} \times \c{Q}$ where $q \notin S$, we can define a rotational grid around $p$, with respect to $q$, so that $S$ is fixed.  Let $\c{R}_{d,S}$ be the subset of rotations in $d$-space under which the set $S$ is invariant.  That is for any $R \in \c{R}_{d,S}$ and any $s \in S$ we have $R(s) = s$.  Let $\tau = d - |S|$.  Then (topologically) $\c{R}_{d,S} = \SO{\tau}$.  
Let $R_{S,p,q} = \min_{R \in \c{R}_{d,S}} ||R(q) - p||$ and let $\hat q = R_{S,p,q}(q)$.  
Let $\c{H}[p,q,S,\eps,\Lambda] \subset \c{R}_{d,S}$ be a set of rotations under which $S$ is invariant with the following property:  for any point $q'$ such that there exists a rotation $R' \in \c{R}_{d,S}$ where $R'(q) = q'$ and where $||q' - \hat q|| \leq \Lambda$, then there exists a rotation $R \in \c{H}[p,q,S,\eps,\Lambda]$ such that $||R(q) - q'|| \leq \eps$.  
For the sanity of the reader, we will not give a technical construction, but just note that it is possible to construct $\c{H}[p,q,S,\eps,\Lambda]$ of size $O((\Lambda/\eps)^{\tau})$.

\paragraph{Algorithm.}
For each pair of ordered sets of $d$ points $(p_1, p_2, \ldots, p_d) \subset \c{P}$ and $(q_1, q_2, \ldots, q_d) \subset \c{Q}$ consider the following set of translations and rotations.  Points in $\c{P}$ may be repeated.  
For each $g \in \c{G}[\eps/d,p_1,\sqrt{\ln(\max\{1/\eps,n^2\})}]$ consider translation $T_{p_1, q_1, g}$ such that $q_1 + T_{p_1, q_1, g} = g$.  We now consider rotations of the set $(\c{Q} \oplus T_{p_1, q_1, g})$.  
Let $S = \{q_1\}$ and consider the rotational grid $\c{H}[p_2, q_2+T_{p_1,q_1,g}, S, \eps/d, \sqrt{\ln(1/\eps)}]$.  For each rotation $R_2 \in \c{H}[p_2, q_2+T_{p_1,q_1,g}, S, \eps/d, \sqrt{\ln(1/\eps)}]$ we recurse as follows.  Apply $R_2(\c{Q} \oplus T_{p_1, q_1, g})$ and place $R_2(q_2 + T_{p_1, q_1, g})$ in $S$.  Then in the $i$th stage consider the rotational grid $\c{H}[p_i, R_{i-1}(\ldots R_2(q_2 + T_{p_1, q_1, g}) \ldots ), S, \eps/d, \sqrt{\ln(1/\eps)}]$.  Where $R_i$ is some rotation we consider from the $i$th level rotational grid, let $\bar{R} = R_d \circ R_{d-1} \circ \ldots \circ R_2$.  Let $(\hat T, \hat R)$ be the pair $(T_{p, q, g}, \bar{R})$ that maximize $\kappa(\c{P},\bar R \circ (\c{Q} \oplus T_{p,q,g}))$.  

\begin{theorem} 
The above algorithm runs in time 
\[
O(n^{2d+2} (1/\eps)^{(d^2-d+2)/2} \log^{(d^2-3d+2)/4} (1/\eps) \log^{d/2} (\max\{n^2, 1/\eps\})),
\]
for a fixed $d$, and is guaranteed to find a translation and rotation pair $(\hat T, \hat R)$, such that 
\[
D_K^2(\c{P}, \hat R \circ (\c{Q} \oplus \hat T)) - D_K^2(\c{P}, R^\star \circ \c{Q} \oplus T^\star) \leq \eps W^2.
\]
\label{lem:rot}
\end{theorem}

\begin{proof}
We compare our solution $(\hat T, \hat R)$ to the optimal solution $(T^\star, R^\star)$.  Note that only pairs of points $(p,q) \in \c{P} \times \c{Q}$ such that $||p - R^\star(q + T^\star)|| < \sqrt{\ln(1/\eps)}$ need to be considered.  

We first assume that for the ordered sets of $d$ points we consider $(p_1, p_2, \ldots, p_d) \subset \c{P}$ and $(q_1, q_2, \ldots, q_d) \subset \c{Q}$ we have 
(A1) $||p_i - R^\star(q_i + T^\star)|| \leq \sqrt{\ln(1/\eps)}$, and 
(A2) for $S = \{q_1, \ldots, q_{i-1}\}$, let $q_i \in \c{Q}$ be the furthest point from $S$ such that $||p_i - (q_i + T^\star)|| \leq \sqrt{\ln(1/\eps)}$.  
Note that (A2) implies that for any rotation $R \in \c{R}_{d,S}$ that $||q_i - R(q_i)|| > ||q' - R(q')||$ for all $q' \in \c{Q}$ that can be within the distance threshold under $(T^\star, R^\star)$.  
In the case that fewer than $d$ pairs of points are within our threshold distance, then as long as these are the first pairs in the ordered sequence, the algorithm works the same up to that level of recursion, and the rest does not matter.  Finally, by Lemma \ref{lem:lb1} we can argue that at least one pair must be within the distance threshold for our transition grid.  

For each point $q \in \c{Q}$ we can show there exists some pair $(T,R)$ considered by the algorithm such that $||R^\star(q + T^\star) - R(q + T)|| \leq \eps.$  
First, there must be some translation $T  = T_{p_1, q_1, g}$ in our grid that is within a distance of $\eps/d$ of $T^\star$.  This follows from Lemma \ref{lem:grid-eps} and similar arguments to the proof for translations.  

For each $q_i$ we can now show that for some $R_i \in \c{H}$ (the rotational grid) we have $||R_i(R_{i-1}(\ldots R_2(q_i + T_{p_1, q_1, g}) \ldots )) - R^\star(q_i + T^\star)|| \leq \eps$.  
By our assumptions, the transformed $q_i$ must lie within the extents of $\c{H}$.  Furthermore, there is a rotation $R_j'$ that can replace each $R_j$ for $j \in [2,i]$ that moves $q_i$ by at most $\eps/d$ such that $R'_i(R'_{i-1}( \ldots R'_2(q_i) \ldots )) = R^\star(q_i)$.  Hence, the composition of these rotations affects $q_i$ by at most $\eps/(i-1)$, and the sum effect of rotation and translation errors is at most $\eps$.  

Since each $q_i$ is invariant to each subsequent rotation in the recursion, we have shown that there is a pair $(T, R)$ considered so $||R(q_i + T) - R^\star(q_i + T^\star)|| \leq \eps$ for $q_i$ in the ordered set $(q_1, q_2, \ldots, q_d)$.
We can now use our second assumption (A2) that shows that at each stage of the recursion $q_i$ is the point affected most by the rotation.  This indicates that we can use the above bound for all points $q \in \c{Q}$, not just those in our ordered set.  

Finally, we can use Lemma \ref{lem:grid-eps} to complete the proof of correctness.  
Since if each $K(p,q)$ has error at most $\eps$, then 
\[
\left|\sum_{p \in \c{P}} \sum_{q \in \c{Q}} \mu(p) \nu(q) K(p, \hat R(q + \hat T)) - \sum_{p \in \c{P}} \sum_{q \in \c{Q}} \mu(p) \nu(q) K(p, R^*(q + T^*))\right| 
\leq
\sum_{p \in \c{P}} \sum_{q \in \c{Q}} \mu(p) \nu(q)  \eps
= 
\eps W^2.
\]

We can bound the runtime as follows.  We consider all $d! {n \choose d} = O(n^d)$ ordered sets of points in $\c{Q}$ and all $n^d$ ordered sets of points from $\c{P}$.  This gives the leading $O(n^{2d})$ term.  
We then investigate all combinations of grid points from each grid in the recursion.  The translation grid has size $O((\Delta/\eps)^d) = O((1/\eps)^d \log^{d/2} (\max\{1/\eps, n^2\}))$.  The size of the $i$th rotational grid is $O((\sqrt{\log(1/\eps)}/\eps)^{d-i}$, starting at $i=2$.  The product of all the rotational grids is the base to the sum of their powers $\sum_{i=1}^{d-1} (d-i) = \sum_{i=1}^{d-1} i = (d-1)(d-2)/2 = (d^2 - 3d +2)/2$, that is $O((1/\eps)^{(d^2 - 3d +2)/2} \log^{(d^2 - 3d +2)/4} (1/\eps))$.  Multiplying by the size of the translational grid we get $O((1/\eps)^{(d^2 - d + 2)/2} \log^{(d^2 - 3d +2)/4} (1/\eps) \log^{d/2} (\max\{n^2, 1/\eps\}))$.  
Then for each rotation and translation we must evaluate $\kappa(\c{P}, R \circ (\c{Q} \oplus T))$ in $O(n^2)$ time.  
Multiplying these three components gives the final bound of
\[
O(n^{2d+2} (1/\eps)^{(d^2 -d +2)/2} \log^{(d^2 -3d +2)/4} (1/\eps) \log^{d/2} (\max\{n^2 , 1/\eps\})).  \qedhere
\]  
\end{proof}


The runtime can again be reduced by first computing a coreset of size $O((1/\eps^2) \log (1/\delta))$ and using this value as $n$.  
After simplifying some logarithmic terms we reach the following result.

\begin{theorem} 
For fixed $d$, in  
\[
O(n + \log n (1/\eps^2)( \log(1/\delta)) + (1/\eps)^{(d^2 + 7d + 6)/2} (\log(1/\eps\delta))^{(d^2 + 7d +10)/4}),
\] 
time we can find a translation and rotation pair $(\hat T, \hat R)$, such that 
\[
D_K^2(\c{P}, \hat R \circ (\c{Q} \oplus \hat T)) - D_K^2(\c{P}, R^\star \circ \c{Q} \oplus T^\star) \leq \eps W^2,
\] 
with probability at least $1-\delta$.
\label{thm:rot}
\end{theorem}

\newpage
\bibliography{refs}
\bibliographystyle{plain}

\appendix


\section{Coresets for the Kernel Distance}
\label{sec:coreset}

\subsection{Alternative Coreset Proof}
In Section \ref{sec:core-set-current} we presented a construction for a coreset for the kernel distance, that depended only on $1/\eps$ for a fixed kernel.  This assumed that $K^+ = \max_{p,q} K(p,q)$ was bounded and constant.  Here we present an alternative proof (also by random sampling) which can be made independent of $K^+$.  This proof is also interesting because it relies on $\rho$ the dimension of the feature map.  

Again,  we sample $k$ points uniformly from $P$ and reweight $\eta(p) = W/k$. 
%

\begin{lemma}
By constructing $S$ with size $k = O((1/\eps^3) \log(n/\delta) \log((1/\eps \delta)\log n))$ we guarantee $D_K^2(\c{P}, \c{S}) \leq \eps W^2$, with probability at least $1-\delta$.  
\label{thm:apx-S}
\end{lemma}

\begin{proof}
The error in our approximation will come from two places: (1) the size of the sample, and (2) the dimension $\rho$ of the feature space we perform the analysis in.  

Let $\phi : \c{P} \times \b{R}^+ \to \c{H}$ describe the true feature map from a point $p \in \c{P}$, with weight $\mu(p)$, to an infinite dimensional feature space.   As before, set $\Phi(\c{P})  = \sum_{p \in \c{P}} \mu(p)\phi(p)$, and recall that $D_K(\c{P},\c{Q}) = \| \Phi(\c{P}) - \Phi(\c{Q})\|_{\c{H}}$, for any pair of shapes $\c{P}$ and $\c{Q}$.  

By the results in the previous section, we can construct  $\tilde \phi : \c{P} \times \b{R}^+ \to \b{R}^\rho$ (such as $\phi_\Upsilon$ defined for Lemma \ref{lem:rand-feat}) such that $\tilde \Phi(\c{P}) = \sum_{p \in \c{P}} \tilde \phi(p)$ and for any pair of shapes $\c{P}$ and $\c{S}$ with weights $W = \sum_{p \in \c{P}} \mu(p) = \sum_{p \in \c{S}} \eta(p)$, we have 
$
\left| \|\Phi(\c{P}) - \Phi(\c{S})\|_{\c{H}}^2 - \|\tilde \Phi(\c{P}) - \tilde \Phi(\c{S})\|^2 \right| \leq (\eps/2) W^2,
$
with probability at least $1-\delta/2$.  
This bounds the error in the approximation of the feature space.

We now use the low dimension $\rho$ of this approximate feature space to bound the sampling error.
Specifically, we just need to bound 
the probability that $\|\tilde \Phi(\c{P}) - \tilde \Phi(\c{S})\|^2  = \|E[\tilde \Phi(\c{S})] - \tilde \Phi(\c{S})\|^2 \geq (\eps/2) W^2$, since $E[\tilde \Phi(\c{S})] = \tilde \Phi(\c{P})$.  This is always true if for each dimension (or pair of dimensions if we alter bounds by a factor $2$) $m \in [1,\rho ]$ we have $\left|\tilde \Phi(\c{S})_m - E[\tilde \Phi(\c{S})_m] \right| \leq \sqrt{\eps W^2/ (2\rho) }$, so we can reduce to a $1$-dimensional problem.

We can now invoke the following Chernoff-Hoeffding bound.  Given a set $\{X_1, \ldots, X_r\}$ of independent random variables, such that $\left| X_i - E[X_i] \right| \leq \Lambda$, then for $M = \sum_{i=1}^r X_i$ we can bound $\Pr[| M - r E[X_i] | \geq \alpha] \leq 2 e^{-2 \alpha^2 / (r\Lambda^2)}$.  

By letting $\alpha = \sqrt{W^2 \eps/(2\rho)}$ and $X_i = \tilde \phi(p_i)_m$, the $m$th coordinate of $\tilde \phi(p_i)$ for $p_i \in S$,   
\begin{align*}
\Pr\left[\left\|\tilde \Phi(\c{S}) - \tilde \Phi(\c{P})\right\|^2 \geq (\eps/2) W^2 \right]
 = & 
\Pr\left[ \left\|\tilde \Phi(\c{S}) - E\left[\tilde \Phi(\c{S})\right]\right\|^2 \geq (\eps/2) W^2\right]
\\ \leq &
\rho  \Pr\left[\left|\Phi(\c{S})_m - k E\left[\tilde \phi(p_i)_m\right]\right| \geq \sqrt{\eps W^2/(2\rho)} \right] 
\\ \leq &
\rho  2 e^{-2\frac{\eps W^2}{2\rho} /(k\Lambda^2)} 
 \leq  
\rho  2 e^{-\frac{\eps W^2}{\rho k} \frac{k^2}{4 W^2}}
 = 
\rho 2 e^{- k \eps /(4\rho)},
\end{align*}
where the last inequality follows because $\Lambda = \max_{p,q \in S} ||\tilde{\phi}(p) - \tilde{\phi}(q)|| \leq 2 W/k$ since for any $p \in S$ we have $||\tilde{\phi}(p)|| = W/k$.  
By setting $\delta/2 \geq \rho 2 e^{-k \eps / (4 \rho)}$, we can solve for $k = O((\rho/\eps) \log (\rho/\delta))$.   
The final bound follows using $\rho = O((1/\eps^2) \log (n/\delta))$ in Lemma \ref{lem:rand-feat}.  
\end{proof}

Again using the feature map summarized by Lemma \ref{lem:rand-feat} we can compute the norm in feature space in $O(\rho k)$ time, after sampling $k = O((1/\eps^3) \log (n/\delta) \log((1/\eps \delta) \log n))$ points from $\c{P}$ and with $\rho = O((1/\eps^2) \log (n/\delta))$.  

\begin{theorem}
We can compute a vector $\hat S = \tilde \Phi(\c{S}) \in \b{R}^\rho$ in time $O(n + (1/\eps^5)\log^2(n/\delta)\log^2((1/\eps\delta)\log n)))$ such that $\left\| \tilde \Phi(\c{P}) - \hat S \right\|^2 \leq \eps W^2$ with probability at least $1-\delta$.  
\end{theorem}

\subsection{NP-hardness of Optimal Coreset Construction}

In general, the problem of finding a fixed-size subset that closely approximates the kernel norm is NP-hard.

\begin{defn}[\textsc{Kernel Norm}]
Given a set of points $\c{P} = \{X_i\}_{i=1}^n$, a kernel function $K$, parameter $k$ and a threshold value $t$, determine if there is a subset of points $\c{S} \subset \c{P}$ such that $|\c{S}|=k$ and $D_K(\c{S},\c{P}) \leq t$.  
\end{defn}

\begin{theorem}
\label{thm:np-hard}
\textsc{Kernel Norm} is NP-hard, even in the case where $k=n/2$ and $t =0$.
\end{theorem}
\begin{proof}
To prove this, we apply a reduction from \textsc{Partition}: given a set $Q = \{x_i\}_{i=1}^n$ of integers with sum to $\sum_{i=1}^n = 2m$, determine if there is a subset adding to exactly $m$. 
Our reduction transforms $Q$ into a set of points $\c{P} = \{x'_i\}_{i=1}^n$ which has subset $\c{S}$ of size $k = n/2$ such that $||\c{S}-\c{P}|| \leq t$ if and only if $Q$ has a partition of two subsets $Q_1$ and $Q_2$ of size $n/2$ such that the sum of integers in each is $m$. 

Let $c = \frac{1}{n} \sum x_i = 2m/n$ and $x'_i=x_{i}-c$ and let $t =0$.  
Let the kernel function $K$ be an identity kernel defined $K(a,b) = \langle a, b \rangle$, where the feature map is defined $\phi(a) = a$.  This defines the reduction.  

Let $s = (n/k) \sum_{x_i' \in \c{S}} x_i'$ and $p = \sum_{x_i' \in \c{P}} x_i'$.  Note that $p=0$ by definition.
Since we have an identity kernel so $\phi(a) = a$, $D_K(\c{S}, \c{P}) = \|s- p\|$.  Thus there exists an $\c{S}$ that satisfies $D_K(\c{S}, \c{P}) \leq 0$ if and only if $s = 0$.  

We now need to show that $s$ can equal $0$ if and only if there exists a subset $Q_1 \subset Q$ of size $n/2$ such that its sum is $m$.  
We can write 
\[
s 
= 
\frac{n}{k}\sum_{x'_i \in \c{S}} x'_i 
= 
2 \sum_{x'_i \in \c{S}} \left(x_i - \frac{2m}{n}\right) 
= 
-2m + 2 \sum_{x_i' \in \c{S}} x_i.
\] 
Thus $s=0$ if and only if $\sum_{x'_i \in \c{S}} x_i = m$.  Since $\c{S}$ must map to a subset $Q_1 \subset Q$, where $x'_i \in \c{S}$ implies $x_i \in Q_1$, then $s=0$ holds if and only if there is a subset $Q_1 \subset Q$ such that $\sum_{x_i \in Q_1} = m$.  This would define a valid partition, and it completes the proof.  
\end{proof}

\end{document}